\documentclass[12pt]{article}

\usepackage{amsmath,amssymb,amsfonts}
\usepackage{tikz,fullpage}
\usetikzlibrary{arrows,petri,topaths}
\usepackage{tikz-network}
\usepackage{fancybox}
\usepackage[noend,noline,ruled,scleft,nofillcomment]{algorithm2e}
\usepackage{caption} 
\usepackage{hyperref}

\newtheorem{theorem}{Theorem}
\newtheorem{lemma}[theorem]{Lemma}

\newtheorem{definition}[theorem]{Definition}
\newtheorem{observation}[theorem]{Observation}

\include{amssymb}

\newenvironment{proof}{\noindent {\sc Proof:}}{$\Box$ \medskip}
\newcommand{\eat}[1]{}
\newenvironment{proofof}[1]{\begin{trivlist} \item {\bf Proof
#1:~~}}
  {\qed\end{trivlist}}

\newcommand{\etal}{{\em et al. }}
\newcommand{\ie}{{\em i.e. }}
\newcommand{\eps}{\varepsilon}

\newcommand{\COMMENTED}[1]{{}}

\newcommand{\poly}{{\mathrm{poly}}}

\begin{document}

\title{An Estimator for Matching Size in Low Arboricity Graphs with Two Applications \footnote{This work is supported by the Iranian Institute for Research in Fundamental Sciences (IPM), Project Number 98050014.}}


\author{Hossein Jowhari \footnote{
Department of Computer Science and Statistics,
Faculty of Mathematics,
K. N. Toosi University of Technology. Email: jowhari@kntu.ac.ir}
}
\maketitle

\begin{abstract}
In this paper, we present a new simple degree-based
estimator for the size of maximum matching in bounded arboricity graphs.
When the arboricity of the graph is bounded by $\alpha$,
the estimator gives a $\alpha+2$ factor approximation
of the matching size. For planar graphs, we show
the estimator does better and returns
a $3.5$ approximation of the matching size. 

Using this estimator, we get 
 new results for approximating the matching size of planar graphs in the streaming and 
distributed models of computation. In particular, 
in the vertex-arrival streams, we get a randomized $O(\frac{\sqrt{n}}{\eps^2}\log n)$ space 
algorithm for approximating the matching size within $(3.5+\eps)$ factor in a 
planar graph on $n$ vertices. Similarly, we get a simultaneous 
protocol in the vertex-partition model for approximating the matching
size within $(3.5+\eps)$ factor using $O(\frac{n^{2/3}}{\eps^2}\log n)$ communication
from each player.  

In comparison with the previous estimators, the estimator in this paper does not need to know the 
arboricity of the input graph and improves the approximation factor
for the case of planar graphs.
\end{abstract}




\section{Introduction}
A matching in a graph $G=(V,E)$ is a subset 
of edges $M \subseteq E$ where no two edges in $M$ share an endpoint. 
A maximum matching of $G$
has the maximum number of edges among all possible matchings.
Here we let $m(G)$ denote the matching size of $G$, \ie the size of a maximum 
matching in $G$. 
In this paper, we present algorithms for approximating $m(G)$
in the sublinear models of computation. In particular, our results
fit the vertex-arrival stream model (also known as the adjacency list 
streams). In the vertex-arrival model, in contrast with the edge-arrival version
 where the input stream
is an arbitrary ordering of the edges, here each item in
the stream is a vertex of the graph followed by a list of its neighbors.

We also focus on graphs with bounded arboricity.
A graph $G=(V,E)$ has arboricity bounded by $\alpha$ if the edge set $E$ 
can be partitioned into at most $\alpha$ forests.
A well-known fact (known as the Nash-William theorem \cite{NW64}) states that a graph has 
arboricity $\alpha$, if and only if every induced subgraph on $t$ vertices 
has at most $\alpha (t-1)$ number of edges. 
Graphs with low arboricity 
cover a wide range of  graphs such as constant degree graphs, planar 
graphs, and graphs with small tree-widths. In particular planar graphs 
have arboricity bounded by $3$. 

A simple reduction from counting distinct elements
implies that computing $m(G)$ exactly
requires $\Omega(n)$ space complexity
even for trees and randomized algorithms
(see \cite{alon1999space} for the lower bound
on distinct elements problem.)
This  has initiated the study of finding 
{\em computationally-light}
estimators for $m(G)$ that take small space to compute.
With this focus,   
following the work by Esfandiari \etal \cite{EHLMO15}, there has been
a series of papers \cite{MV16,CormodeJMM17,MV18,BuryGMMSVZ19}
 that have designed estimators for
the matching size based on the degrees of vertices, edges
and the arboricity of the input graph. 
In this paper, we design another degree-based estimator for 
$m(G)$ in low arboricity graphs that has
certain advantages in comparison with the previous works
and leads to new algorithmic results. Before describing
our estimator we briefly review some of the previous ideas.   
In the discussions below, we assume $G$
has arboricity bounded by $\alpha$.

\paragraph*{Shallow edges, high degree vertices} 
Esfandiari \etal \cite{EHLMO15}
 were first to observe that one can approximately 
 characterize the matching size of low arboricity graphs
 based on the degree information of the vertices and the local neighborhood
 od the edges. 
 Let $H$ denote the set of vertices with degree more than
  $h=2\alpha+3$ and let $F$ denote the set of edges with both
   endpoints having degree
at most $h$.  Esfandiari \etal have shown that 
$m(G) \le |H|+|F| \le (5\alpha+9)m(G).$ 
 Based on this estimator,  the authors in \cite{EHLMO15} have designed a $\tilde{O}(\eps^{-2}\alpha n^{2/3})$
  space algorithm for approximating $m(G)$ within
  $5\alpha+9+\eps$ factor in the edge-arrival model. 
 
 \paragraph*{Fractional matchings} 
 By establishing an interesting connection with fractional matchings
 and the Edmonds Polytope theorem,
 Mcgregor and Vorotnikova \cite{MV16} have shown the following quantity
 approximates
 $m(G)$ within $(\alpha+2)$ factor. 
 $$(\alpha+1)\sum_{(u,v) \in E }\min \{\frac1{\deg(u)}, \frac1{deg(v)},\frac1{\alpha+1}\}.$$
  Based on this estimator, the authors in \cite{MV16} have given a $\tilde{O}(\eps^{-2}n^{2/3})$ 
 space streaming
 algorithm (in the edge-arrival model) that approximate $m(G)$ within
 $\alpha+2+\eps$ factor. Also in the same work, 
 another degree-based estimator is given that
 returns a $\frac{(\alpha+2)^2}2$ factor approximation of $m(G)$.
 A notable property of this estimator is that it can be implemented
 in the vertex-arrival stream model in $O(\log n)$ bits of space.
    
  \paragraph*{$\alpha$-Last edges} Cormode \etal \cite{CormodeJMM17} 
  (later revised by Mcgregor and Vorotnikova \cite{MV18}) have designed an
   estimator that depends on a given
  ordering of the edges. Given a stream of edges $S=e_1,\ldots,e_m$, 
  let $E_\alpha(S)$ denote a subset of edges
where $(u,v)\in E_\alpha(S)$ iff the vertices $u$ and $v$ both appear at most 
$\alpha$ times in $S$ after the edge $(u,v)$. 
It is shown that $m(G) \le |E_\alpha(S)| \le (\alpha+2)m(G).$
Moreover a $O(\frac1{\eps^2} \log^2n)$ space 
streaming algorithm is given that approximates $|E_\alpha(S)|$ within $1+\eps$ factor
in the edge-arrival model.  

\subsection{The estimator in this paper} The new estimator
 is purely based on the degree of the 
vertices in the graph without any dependence on $\alpha$. 
To estimate the matching size, we count the 
number of what we call {\em locally superior} vertices in the graph. Namely,

\begin{definition}
In graph $G=(V,E)$, we call $u \in V$ a locally superior vertex
 if $u$ has a 
neighbor $v$ such that $deg(u) \ge deg(v)$. We let $\ell(G)$ 
denote the number of locally superior vertices in $G$.   
\end{definition}
 
We show if the arboricity of $G$ is bounded by $\alpha$, then 
$\ell(G)$ approximates $m(G)$ within $(\alpha+2)$ factor
(Lemma \ref{lem:locally superior}.)
This repeats the same bound obtained by the estimators 
in \cite{MV16} and \cite{CormodeJMM17},
 however for planar graphs, we show that the approximation factor is 
 at most $3.5$ which beats the previous bounds (Lemma \ref{lem:planar})
 \footnote{We do not have tight examples for our analysis.
 In fact, we conjecture that $\ell(G)$ approximates $m(G)$
 within $3$ factor when $G$ is planar.}.
  As an evidence, consider the 4-regular planar
graph on $9$ vertices.  Both of the estimators in 
\cite{MV16} and \cite{CormodeJMM17},
report $18$ as the estimation for $m(G)$ while the exact answer
is $4$. It follows their approximation factor is at least $4.5$. 

Unfortunately,  the new estimator, in spite of
its simplicity, does not seem to be applicable 
in the edge-arrival
model without an extra pass over the stream. 
To decide if a vertex is locally superior, we
need to know its neighbors and learn their degrees which 
becomes burdonsome in one pass. However in the 
vertex-partition model, we 
can obtain this information in one pass and consequently
can achieve sublinear space bounds. More formally, we
get a randomized $O(\frac{\sqrt{n}}{\eps^2}\log n)$ space
algorithm for approximating $m(G)$ within $(3.5+\eps)$ factor
in this model. In terms of
approximation factor, this improves over existing
sublinear algorithms \cite{MV16,MV18}.

As another application of our estimator, we get a sublinear simultaneous
 protocol in the {\em vertex-partition} 
model for approximating 
$m(G)$ when $G$ is planar. In this model, vertex set $V$ is partitioned
into $t$ subsets $V_1,\ldots,V_t$ where each subset is given to a player.
 The $i$-th player knows the edges on $V_i$. 
 The players do not communicate with each other. They only send one message to a {\em referee}
  whom at the end computes an approximation of the matching size. (The referee does not get any part of the input.)
  We assume the referee and the players have a shared source of randomness. Within this setting, we design a protocol that approximates 
  $m(G)$ within $3.5+\eps$ factor using $O(\frac{n^{2/3}}{\eps^2}\log n)$ communication
  from each player. 
  Note that for $t > 3$ and $t = o(n^{1/3})$, this result
is non-trivial. The best previous result implicit in the works of \cite{CCEHMMV16, MV16} computes
a  $5+\eps$ factor approximation using $\tilde{O}(n^{4/5})$ communication
from each player. We should also mention that, based on
the estimator in \cite{MV16}, there is a simultaneous 
protocol that reports a $12.5$ factor approximation
of $m(G)$ using $O(\log n)$ communication.

\section{Graph properties}
\label{sec:est}
In the following proofs, we let $M \subseteq E$ denote a maximum matching in graph $G$. 
When the underlying graph is clear from the context,
for the vertex set $S$,
we use $N(S)$ to denote the neighbors of the vertices in $S$ excluding $S$ itself.
For vertex $u$, we simply use $N(x)$ to denote the neighbors of $u$. 
The vertex $v$ is a neighbor of
the edge $(x,y)$ if $v$ is adjacent with $x$ or $y$. 
When $x$ is paired with $y$ in the matching $M$, 
abusing the notation, we define $M(x)=y$. 

\begin{lemma} 
\label{lem:locally superior}
Let $G=(V,E)$ be a graph with arboricity $\alpha$. We have
$$m(G) \le \ell(G) \le (\alpha + 2) m(G).$$
\end{lemma}

\begin{proof} 
The left hand side of the inequality is easy to show. For every edge in $E$, at least one
of the endpoints is locally superior. Since edges in $M$ are disjoint, 
at least $|M|$ number of endpoints must be locally superior. This proves $m(G)\le \ell(G)$. 

To show the right hand side, we use a charging argument.
Let $L$ denote the locally superior vertices in $G$. Our goal is to show an upper bound on $|L|$ in terms 
of $|M|$ and $\alpha$.
Let $X \subseteq L$ be the set of 
locally superior vertices that are NOT endpoints of a matching edge. 
The challenge is to prove an upper on $|X|$. 

 The vertices in $X$ do not contribute to the maximum 
matching. However all the vertices in $N(X)$ must be endpoints of 
matching edges (otherwise $M$ would not be a maximal matching.)
For the same reason, there cannot be an edge between the vertices in $X$. 
To prove an upper bound on $|X|$, in the first step,
using an assignment procedure, we assign a subset of
 vertices in $X$
to edges in $M$ in a way any target edge gets at most $\alpha-1$ 
locally superior vertices.  We do the assignments
in the following way. 

\paragraph*{The Assignment Procedure}
If we find a $y \in N(X)$ with at most $\alpha-1$ neighbors
in $X$, we assign all the neighbors of $y$ in $X$ to the matching edge
$(y,M(y))$.
 We repeat this process, every time 
picking a vertex in $N(X)$ with  less than 
$\alpha$ neighbors in $X$ and do the assignment 
that we just described,
 until we cannot find such a 
vertex in $N(X)$.
   Note that when we assign a locally superior vertex 
   $x$, we remove the edges on $x$ before continuing the procedure.

\vspace{0.5cm}
Here we emphasize the fact that if $y$ has a neighbor $x \in X$, then $M(y)$ cannot have neighbors in $X \setminus \{x\}$ (otherwise it would create an augmenting path and contradict the optimality of 
$M$.)

\vspace{0.5cm}
  Let $X_1 \subseteq X$ be the assigned locally superior vertices and
  $M_1 \subseteq M$ be the used matching edges in the assignment procedure.
   We have 
\begin{equation}
\label{eqn:X1}
   |X_1| \le (\alpha-1)|M_1|.  
\end{equation}   
%
%
\vspace{0.5cm}
Let $X_2 = X \setminus X_1$ be the unassigned vertices in $X$.
Now we try to prove an upper bound on $|X_2|$. For this, we need
to make a few observations. 
 \begin{observation}
 \label{obs:Y}
  Let $Y_2=N(X_2)$. 
 The pair $y$ and $M(y)$ cannot be both in $Y_2$. 
 \end{observation}
\begin{proof}
Suppose $y$ and $M(y)$ are both in $N(X_2)$.
Let $B$ and $C$ be the neighbors of $y$ and $M(y)$ in $X_2$ 
respectively. If $|B \cup C| > 1$, 
then one can find an augmenting path of length $3$ (with respect to $M$.)
 A contradiction.

On the other hand, if $|B \cup C| =1$, then
$y$ and $M(y)$ have only a shared neighbor $x \in X_2$ which means the edge 
$e=(y,M(y))$ should have been used by the assignment procedure and
as result $x \in X_1$. Another contradiction.
\end{proof}

 \begin{observation}
 \label{obs:X2degree}
Every vertex $x \in X_2$ has degree at least $\alpha+1$.
 \end{observation}
 \begin{proof}
 Consider $x \in X_2$. Suppose, for the sake 
 of contradiction, $\deg(x)$ 
 is $k$ where $k \le \alpha$. Since
 $x$ is a locally superior vertex, there must be a $y \in N(x)$ with degree
 at most $k$ in $G$. We know that $y$ is an endpoint of a matching edge. 
 In the assignments procedure, whenever we used an edge $e \in M$
 all the neighbors of its endpoints (in $X$) were assigned. Since $x$ is not 
 assigned yet,
 it means the edge $(y,M(y))$ has not been used. Consequently 
 $y$ must have at least $\alpha$ neighbors in $X_2$. Counting the edge
 $(y,M(y))$, we should have $\deg(y) \ge \alpha+1$. A contradiction. 
 \end{proof}

  Let $G'=(X_2 \cup Y_2,E')$ be a bipartite graph where $E'$ is the set
  of edges between $X_2$ and $Y_2$.  
 From Observation \ref{obs:X2degree}, we have

\begin{equation}
\label{eq:lowerboundE'}
  (\alpha+1)|X_2|\le |E'|. 
\end{equation}

Since $G'$ is a subgraph of $G$, its arboricity is bounded by $\alpha$. As result,

\begin{equation}
\label{eq:upperboundE'}
 |E'| \le \alpha(|X_2|+|Y_2|). 
\end{equation}

Recall that $Y_2$ are endpoints of matching edges.  Let $M_2$ be those matching edges. 
Observation \ref{obs:Y} implies that $|Y_2|=|M_2|$.
As result, combining (\ref{eq:lowerboundE'}) and (\ref{eq:upperboundE'}), we get the following.

\begin{equation}
 |X_2| \le \alpha|Y_2|=\alpha|M_2|.
\end{equation}
To prove an upper bound on $|L|$, we also need to count the locally superior vertices
that are endpoints of matching edges. Let 
 $Z  = L\setminus X$. We have 
$|Z|  \le 2|M|.$ Summing up, we get
\begin{align*}
|L|&=|X_1| + |X_2| + |Z| \\
        & \le (\alpha -1)|M_1| + \alpha|M_2|  + 2|M|  \\
         & = \alpha(|M_1| +|M_2|) + 2|M|-|M_1|  \\
                  & \le (\alpha +2)|M| - |M_1|  \\
                                    & \le (\alpha +2)|M| 
\end{align*}
This proves the lemma.
\end{proof}

\begin{lemma} 
\label{lem:planar}
Let $G=(V,E)$ be a planar graph.  We have $ \ell(G) \le 3.5 m(G)
$.
\end{lemma}
\begin{proof}
For planar graphs, similar to what we did
in the proof of Lemma \ref{lem:locally superior}, we first try
to assign some of the vertices in $X$ to the matching edges using a simple 
assignment procedure. (Recall that $X$ is the set of 
vertices in $L$ that are not endpoints of edges in $M$.) 

\paragraph{The Assignment Procedure} Let $Y_1 =\emptyset$.
If we find a $y \in N(X)$ with only $1$ neighbor $x \in X$, we assign 
$x$ to the matching edge
$(y,M(y))$. Also we add $y$ to $Y_1$.
 We continue the procedure until we cannot find such a 
vertex in $N(X)$.
   Note that when we assign a locally superior vertex $x$, we remove the edges on $x$.

\vspace{0.5cm}
  Let $X_1 \subseteq X$ be the assigned locally superior vertices and
  $M_1 \subseteq M$ be the used matching edges in the assignment procedure. Note that $|Y_1|=|M_1|$.
   We have 
\begin{equation}
\label{eqn:planarX1}
   |X_1| \le |M_1|.  
\end{equation}  

Let $X_2 = X \setminus X_1$.
Using a similar argument that we used for proving Observation \ref{obs:X2degree},
we can show every vertex in $X_2$ has degree at least $3$. 
Also letting $Y_2=N(X_2)$, we observe that $y \in Y_2$ and $M(y)$ cannot
be both in $Y_2$ as we noticed in the Observation \ref{obs:Y}. Let $M_2 \subseteq M$ be the matching 
edges with one endpoint
in $Y_2$. We have $|Y_2|=|M_2|$.
 
 Now consider the bipartite graph $G'=(X_2 \cup Y_2,E')$ where $E'$ is the set
  of edges between $X_2$ and $Y_2$. Every planar bipartite graph with $n$ vertices has
at most $2n-4$ edges {\footnote {For a short proof of this, combine the 
Euler's formula $|V|-|E|+|F|=2$
with the inequality $2|E|\ge 4|F|$ caused by each face having at least 4 
sides (since there are no odd cycles) and we get $|E|\le 2|V|-4.$}}. 
Since $G'$ is a bipartite planar graph, it follows,
  
  \begin{equation}
\label{eqn:planarE'}
   3|X_2| \le |E'| < 2(|X_2|+|Y_2|)= 2(|X_2|+|M_2|).  
\end{equation}  

This shows $|X_2| < 2|M_2|$.
Letting $Z=L \setminus X$ and $M_3=M\setminus(M_1 \cup M_2)$, we get 

\begin{equation}
\label{eqn:R_upperbound1}
|L|=|X_1| + |X_2| + |Z| \le |M_1| + 2|M_2|  + 2|M|  
         \le 3|M| +|M_2| -|M_3|. 
\end{equation}  

This already proves $|L|$ is bounded by $4|M|$. To prove the bound claimed in the lemma,
 we also show that $|L| \le 3|M|+|M_1|+|M_3|$. Combined with the
  inequality (\ref{eqn:R_upperbound1}),
this proves the lemma. 

\vspace{0.5cm}
Let $Y=Y_1 \cup Y_2$. Note that $Y$ are one side of
 the matching edges in $M_1 \cup M_2$. Let
$Y'=\{ M(y) \; | \; y \in Y\}.$ 
We use a special subset
of $Y'$, named $Y''$ which is defined as follows. 
We let $Y''$ denote the locally superior vertices in $Y'$ that have degree $2$ or they
are adjacent with both endpoints of an edge in $M_3$. 
We make the following observation regarding the vertices
in $Y''$.

\begin{observation}
\label{obs:Y''}
We can assign each vertex $y' \in Y''$ to a distinct 
$e \in Y_1 \cup M_3$ where
$e$ has no neighbor in $Y'\setminus \{y'\}$.
\end{observation}

\begin{proof}
Consider $y' \in Y''$. If $y'$ is adjacent with 
both endpoints of an edge $e=(z,z') \in M_3$, we assign
$y'$ to $e$ (when there are multiple edges with this
condition we pick one of them arbitrarily.) Note that $z$ and $z'$ cannot have neighbors in $Y'$
other than $y'$ because otherwise it would create an
augmenting path.  

Now suppose $y'$ has degree $2$. Since $y'$ is a locally superior vertex,
it must have a neighbor $z$ of degree at most $2$. The 
neighbor $z$ cannot be in $Y_2 \cup X_2$ because
the vertices in $Y_2 \cup X_2$ have degree at least $3$.
 We distinguish between two cases.
\begin{itemize}

\item $M(y') \in Y_2$. In this case, $z$ cannot be in $Y_1$ either 
because the vertices in $Y_1$ are already
of degree $2$ without $y'$. Also $z \notin X_1$ because otherwise
 it would create an augmenting path.
The only possibility is that $z$ is an endpoint of a matching
 edge in $M_3$. We assign
$y'$ to the matching edge $(z,z') \in M_3$.  Note that
 $z'$ cannot have a neighbor in $Y'\setminus \{y'\}$
because it would create an augmenting path. 

\item $M(y') \in Y_1$. Here $z$ could be in $X_1$. 
If this is the case, then $M(y')$ cannot have
a neighbor in $Y' \setminus \{y'\}$ because it would create an 
augmenting path. In this case, 
we assign $y'$ to $M(y')$. If $z = M(y')$, then again 
we assign $y'$ to $M(y')$. The only
remaining possibility is that $z$ an endpoint of a 
matching edge in $M_3$ which we handle it similar to
the previous case.  
\end{itemize}
\end{proof}

 \begin{figure}
\centering
\begin{tikzpicture}[scale=0.5,transform shape]

  \Text[x=-1.8,y=0]{\Large $Y_2$}         
  \Vertex[x=0,y=0]{Y21}
  \Vertex[x=2.5,y=0]{Y22}
  \Vertex[x=5,y=0]{Y23}
  \Vertex[x=7.5,y=0]{Y24}
  
    \Text[x=20.3,y=0]{\Large $Y_1$}         
 \Vertex[x=11,y=0]{Y11}
\Vertex[x=13.5,y=0]{Y12}
  \Vertex[x=16,y=0]{Y13}
  \Vertex[x=18,y=0]{Y14}
  
    \Text[x=-2.2,y=2]{\Large $Y'$}         
  \Vertex[x=0,y=2]{Y'21}
  \Vertex[x=2.5,y=2]{Y'22}
  \Vertex[x=5,y=2,color=white]{Y'23}
  \Vertex[x=7.5,y=2]{Y'24}
  
 \Vertex[x=11,y=2,color=white]{Y'11}
\Vertex[x=13.5,y=2]{Y'12}
  \Vertex[x=16,y=2]{Y'13}
  \Vertex[x=18.5,y=2,color=white]{Y'14}
   \Text[x=-1.8,y=-2.3]{\Large $X_2$}         
  \Vertex[x=0,y=-2.3]{X21}
  \Vertex[x=2.5,y=-2.3]{X22}
  \Vertex[x=5,y=-2.3]{X23}
  \Vertex[x=7.5,y=-2.3]{X24}
 
    \Text[x=20.3,y=-2.3]{\Large $X_1$}         
 \Vertex[x=11,y=-2.3]{X11}
\Vertex[x=13.5,y=-2.3]{X12}
  \Vertex[x=16,y=-2.3]{X13}
  \Vertex[x=18.5,y=-2.3]{X14}

    \Text[x=-2.3,y=4]{\Large $M_3$}           
    \Vertex[x=5,y=4]{M31}
  \Vertex[x=7,y=4]{M32}
  \Vertex[x=14,y=4]{M33}
 \Vertex[x=16,y=4]{M34}
  \Vertex[x=9,y=4]{M35}
    \Vertex[x=11,y=4]{M36}

    \draw (3.7,0) ellipse (5cm and 1cm);
        \draw (14.8,0) ellipse (5cm and 1cm);
                \draw (3.7,-2.2) ellipse (5cm and 1cm);
        \draw (14.8,-2.2) ellipse (5cm and 1cm);
      \draw (9.2,2) ellipse (11cm and 0.7cm);
      \draw (9.2,4) ellipse (11cm and 0.7cm);

        \draw (4.8,-0.17) -- (4,-0.8);
        \draw (5.2,-0.17) -- (6,-0.8);

     \draw (2.4,-2) -- (1.8,-1.5);
     \draw (2.5,-2) -- (2.5,-1.5);
     \draw (2.6,-2) -- (3.2,-1.5);

\Edge[lw=4pt](M31)(M32)
\Edge[lw=4pt](M33)(M34)
\Edge[lw=4pt](M35)(M36)

\Edge[lw=4pt](Y11)(Y'11)
\Edge[lw=4pt](Y12)(Y'12)
\Edge[lw=4pt](Y13)(Y'13)
\Edge[lw=4pt](Y14)(Y'14)
\Edge[lw=4pt](Y21)(Y'21)
\Edge[lw=4pt](Y22)(Y'22)
\Edge[lw=4pt](Y23)(Y'23)
\Edge[lw=4pt](Y24)(Y'24)
\Edge[lw=1pt](M36)(Y'11)

\Edge[lw=1pt,style=dashed](Y'13)(Y14)

\Edge[lw=1pt](M31)(Y'23)
\Edge[lw=1pt](M33)(Y'12)
\Edge[lw=1pt](M33)(Y'13)

\Edge[lw=1pt](Y11)(X11)
\Edge[lw=1pt](Y12)(X12)
\Edge[lw=1pt](Y13)(X13)
\Edge[lw=1pt](Y14)(X14)
\Edge[lw=1pt](Y'12)(Y11)
\Edge[lw=1pt](Y'13)(Y12)
\Edge[bend=35, lw=1pt](Y'14)(X14)

\end{tikzpicture}
\caption{
A demonstration of the construction in the proof of lemmas \ref{lem:locally superior} and \ref{lem:planar}. Thick
edges represent matching edges. The unfilled vertices belong to the set $Y''$.  
}
\label{fig_graph}
\end{figure}
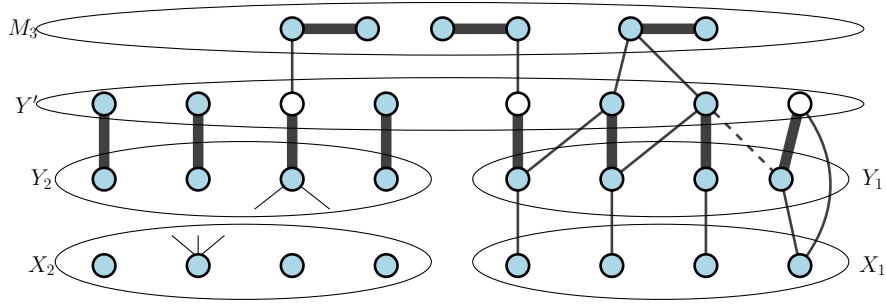

Now, assume we assign the vertices in $Y''$ to the elements in 
$Y_1 \cup M_3$ according to the above
observation. Let $Y_{1}' \subseteq Y_1$ and $M_{3}' \subseteq M_3$ be the vertices and edges 
that were used in the assignment. 
Let $Y'''$ be the remaining locally superior vertices in $Y'$. Namely, $Y'''= (L\cap Y')\setminus Y''$.
Before making the final point, we observe that only
one endpoint of the edges in $M_3$ are adjacent with vertices in $Y'''$. 
Let $Y_3$ be the endpoint of edges in $M_3 \setminus M_3'$ that have
neighbors in $Y'''$.
  Consider the bipartite graph $G''(V'',E'')$
where $$V''=(X_2 \cup Y''') \cup \big(Y_2  \cup (Y_1 \setminus Y_{1}')  \cup Y_3\big)$$
and $E''$ is the set of edges between $X_2$ and $Y_2$, and the edges
between $Y'''$ and $ Y_2 \cup (Y_1 \setminus Y_1')\cup Y_3$.

Relying on the facts that $G''$ is a planar bipartite graph, $Y'''$ is composed
of vertices with degree at least $3$, and the edges on $Y'''$ are all
in $E''$, we have
$$ 3|X_2| + 3|Y'''| \le |E''| \le 2 (|X_2|+|Y_2|+|Y_1\setminus Y_1'|+|Y'''|+|Y_3|).$$
It follows,
\begin{align*}
|X_2| + |Y'''|  & \le 2(|Y_2|+|Y_1\setminus Y_1'|+|Y_3|) \\
                & \le 2(|M_2|+|M_1|-|Y_1'|+|M_3|-|M_3'|) \\
                & = 2(|M|-|Y_1'|-|M_3'|) 
\end{align*}
Since $|Y''|=|Y'_1|+|M'_3|$, we get
\begin{equation}
\label{eq:Y'''}
|X_2| + |Y'''|  \le 2|M|-2|Y''|
\end{equation}

Let $Z_1$, $Z_2$ and $Z_3$ denote the locally superior vertices that
are endpoints of matching edges in $M_1$, $M_2$ and $M_3$ respectively.
From the definition of $Y''$ and $Y'''$, we have
\begin{equation}
\label{eq:Z1Z2}
|Z_1| + |Z_2| \le |M_1| + |M_2| + |Y''| +|Y'''|
\end{equation}
From (\ref{eq:Y'''}) and (\ref{eq:Z1Z2}), we get
\begin{align*}
|L|  & = |X_1| + |X_2| + |Z_1| + |Z_2| + |Z_3| \\
     & \le |M_1| + |X_2| + (|M_1| + |M_2| + |Y''| +|Y'''|) + 2|M_3|\\
          & = 2|M_1| + (|X_2|+|Y'''|) + |M_2| + |Y''| + 2|M_3|\\
          & \le 2|M_1|  + |M_2| + 2|M| - |Y''| + 2|M_3|\\
          & = 3|M|+ |M_1| + |M_3| -|Y''|\\
          & \le 3|M|+ |M_1| + |M_3| 
\end{align*}
This
finishes the proof 
of the lemma. 
\end{proof}

\section{Algorithms}
\label{sec:alg} 
We first present a high-level sampling-based estimator 
for $\ell(G)$. Then we show how this estimator can
be implemented in the streaming and distributed settings
 using small space and communication. For our
 streaming result, we use a combination of the estimator 
 for $\ell(G)$ and the greedy maximal matching algorithm.
For the simultaneous protocol, we use
the estimator for $\ell(G)$ in combination with the 
edge-sampling primitive in \cite{CCEHMMV16}
and an estimator in \cite{MV16}.

 The high-level estimator (described in Algorithm \ref{alg:locally superior}) samples a 
 subset of vertices $S \subseteq V$ and computes the locally superior vertices in $S$. 
 The quantity $\ell(G)$ is estimated from the scaled 
 ratio of the locally superior vertices in the sample set. 

\begin{algorithm}[th]

\medskip

Run the following estimator $r=\lceil \frac 8{\epsilon^2}\rceil $ number of 
times in parallel. In the end, report the average of the outcomes. \medskip

\hspace{1cm} 1. Sample $s$ vertices (uniformly at random) from $V$ without replacement. 

\hspace{1cm} 2. Let $S$ be the set of sampled vertices. 

\hspace{1cm} 3.  Compute $S'$ where $S'$ is the set of locally superior vertices in $S$.

\hspace{1cm} 4.  Return $\frac{n}{s}|S'|$ as an estimation for $\ell(G)$.
 \caption{The high-level description of the estimator for $\ell(G)$}
 \label{alg:locally superior}
\end{algorithm}

\begin{lemma} 
\label{lem:locally superioralg}
Assuming $s \ge \frac{n}{\ell(G)}$, the high-level estimator in Algoirthm \ref{alg:locally superior}
  returns a $1+\eps$ factor approximation of $\ell(G)$ with probability 
  at least $7/8$.
\end{lemma}

\begin{proof}
 Fix a parallel repetition of the algorithm and let $X$ denote the outcome of the associated estimator. 
Assuming an arbitrary ordering on the locally superior vertices, 
let $X_i$ denote the random variable associated with $i$-th locally superior vertex. 
We define $X_i=1$ if the $i$-th locally superior vertex has been sampled, otherwise $X_i=0$. 
We have $X = \frac{n}{s} \sum_{i=1}^{\ell(G)} X_i$.  Since $Pr(X_i=1)=\frac{s}{n}$, we get $E[X]=\ell(G)$.
Further we have
\begin{align*}
E[X^2]= \frac{n^2}{s^2}E\Big[\sum_{i,j}^{\ell(G)}X_iX_j\Big] & = 
\frac{n^2}{s^2}\Big[\sum_i^{\ell(G)}E[X_i^2]+\sum_{i\neq j}^{\ell(G)}E[X_iX_j]\Big] \\
 & = \frac{n^2}{s^2}\Big[\frac{s}{n}\ell(G)+{\ell(G) \choose 2}\frac{s(s-1)}{n(n-1)}\Big] \\
  & = \frac{n}{s}\ell(G)+{\ell(G) \choose 2}\frac{n(s-1)}{s(n-1)} \\
 & < \frac{n}{s}\ell(G) + \ell^2(G)
\end{align*}
 
Consequently,
$Var[X]=E[X^2]-E^2[X]< \frac{n}{s}\ell(G).$ 

Let $Y$ be the average of the outcomes of $r$ parallel and independent repetitions of the 
basic estimator. We have $E[Y]=\ell(G)$ and $Var[Y]<\frac{n}{sr}\ell(G)$. 
Using the Chebyshev's inequality,
$$Pr(|Y-E[Y]|\ge \eps E[Y])\le \frac{Var[Y]}{\eps^2E^2[X]}<\frac{n/s}{r\eps^2 \ell(G)}.$$
Setting $r=\frac8{\eps^2}$ and $s \ge \frac{n}{\ell(G)}$, the above 
probability will be less than $1/8$.
\end{proof}

\subsection{The streaming algorithm}
We first note that we can implement the high-level estimator of
Algorithm \ref{alg:locally superior} in
the vertex-arrival stream model using $O(\frac{s}{\eps^2}\log n)$ space.
Consider a single repetition of the estimator. 
The sampled set $S$ is selected in the beginning
of the algorithm (before the stream.) This can be done
using a reservoir sampling strategy \cite{Vitter85} in $O(|S|\log n)$ space. 
To decide if $u \in S$ is locally superior or not,
we just need to store $\deg(u)$ and 
 the minimum degree of the neighbors that are visited so far.
This takes $O(\log n)$ bits of space. As result,
the whole space needed to implement a single repetition is $O(s\log n)$ bits. 
 
The streaming algorithm runs two threads in parallel. 
In one thread it runs the streaming implementation
of Algorithm \ref{alg:locally superior} after setting $s=\lceil \sqrt{n} \:\rceil$.
In the other thread, it runs a greedy algorithm to find 
a maximal matching in the input graph. We stop the greedy 
algorithm whenever the size of the discovered matching $F$ 
exceeds $\sqrt{n}$. In the end, if $|F| < \sqrt{n}$, we output $|F|$ as an approximation for $m(G)$, otherwise we report the outcome of the first thread. 

Note that if $|F| < \sqrt{n}$, $F$ is a maximal matching in $G$. Hence
$|F| \ge \frac12 m(G)$. Assume $|F| \ge \sqrt{n}$. 
In this case the algorithm outputs the result of first thread. In this case, by Lemma
\ref{lem:locally superior}, we know $\ell(G) \ge \sqrt{n}$. Consequently, it follows from
Lemma \ref{lem:locally superioralg}, the first thread returns a $1+O\eps$ approximation 
of $\ell(G)$ and consequently it returns a $3.5+O(\eps)$ approximation of $m(G)$. 
Since the greedy algorithm takes
at most $O(\sqrt{n})$ space, the space complexity of the algorithm is 
dominated by the space usage of the first thread.
We get the following result. 
   
\begin{theorem} 
\label{thm:planar}
Let $G$ be a planar graph. There is a streaming algorithm 
(in the vertex-arrival model) that returns a $3.5+\epsilon$ factor approximation of 
$m(G)$ using $O(\frac{\sqrt{n}}{\epsilon^2})$ space.
\end{theorem}

%
\subsection{A simultaneous communication protocol}
To describe the simultaneous protocol, we consider two cases separately: (a) 
when the matching size is low; to be precise, when it is smaller than some fixed 
value $k=n^{1/3}$, and (b) when the matching size is high, \ie at least 
$\Omega(k)$. For each case, we describe a separate solution. The overall protocol 
will be these solutions (run in parallel)  
combined with a sub-protocol (in parallel) to distinguish between the cases. 

\paragraph*{Graphs with large matching size}
In the case when matching size is large, similar to what was done in the 
streaming model, we run an implementation of Algorithm 1
in the simultaneous model. To see the implementation, in the 
simultaneous model all 
the players (including the referee) know the sampled set $S$. This results
from access to the shared randomness. For 
each $u \in S$, the players send the minimum degree of the neighbors
of $u$ in his input to the referee. The player
that owns $u$, also sends $\deg(u)$ to the referee.
 Having received this information, the referee can decide
 if $u$ is a locally superior vertex or not. 
 As result, we can implement Algorithm \ref{alg:locally superior} in
 the simultaneous model using a protocol with 
 $O(\frac{s}{\eps^2}\log n)$ message size.

%
\paragraph*{Graphs with small matching size}
In the case where the matching size is small, we use the edge-sampling method
of \cite{CCEHMMV16}. We review their
basic sampling primitive in its general form. 
Given a graph $G(V,E)$,
let $c:V \rightarrow [b]$ be 
a totally random function that assigns each vertex in $V$ 
a random number (color) in $[b]=\{1,\ldots,b\}$. 
The set $\textrm{Sample}_{b,d,1}$ is a random subset of $E$ 
picked in the following way. Given a subset $K \subseteq [b]$ of size $d \in \{1,2\}$, 
let $E_K$ be the edges of $G$ where the color of their endpoints 
matches $K$. For example when $K=\{3,4\}$, the set $E_{\{3,4\}}$ 
contains all edges $(u,v)$ such that $\{c(u),c(v)\}=\{3,4\}$. For 
all $K \subseteq [b]$ of size $d$, the set $\textrm{Sample}_{b,d,1}$ 
picks a random edge from $E_{K}$. Finally, 
the random set $\textrm{Sample}_{b,d,r}$ is the union of $r$ independent 
instances of $\textrm{Sample}_{b,d,1}$. 
We have the following lemma from \cite{CCEHMMV16} (see
Theorems 4 in the reference.)

\begin{lemma} 
\label{lem:Sample}
Let $G=(V,E)$ be a graph. 
When $m(G) \le k$, with probability $1-1/\poly(k)$, 
the random set $\textrm{Sample}_{100k,2,O(\log k)}$  contains a matching 
of size $m(G)$. 
\end{lemma}

Note that, in the simultaneous vertex-partition model,
the referee can obtain an instance
of $\textrm{Sample}_{b,d,1}$ 
 via a protocol with
$O(b^d\log n)$ message size. To see this, using the shared 
randomness, the players pick the random function $c: V \rightarrow [b]$. 
Let $E^{(i)}$ be the subset of edges owned by the $i$-th player. 
We have $E=\bigcup_{i=1}^t E^{(i)}$.  To pick a random edge 
from $E_K$ for a given $K \subseteq [b]$, 
the $i$-th player randomly picks an edge  $e \in E_K \cap E^{(i)}$ 
and sends it along with $|E_K \cap E^{(i)}|$ to the referee. After 
receiving this information from all the players, the referee can 
generate a random element of $E_K$.  Since there 
are $O(b^d)$ different  $d$-subsets of $[b]$, the size
of the message from a player to the referee 
is bounded by $O(b^d\log n)$ bits. 
Consequently, the referee can produce a rightful instance of $\textrm{Sample}_{b,d,r}$
using $O(rb^d\log n)$ communication from each player. 

\paragraph*{How to distinguish between the cases?}
To accomplish this, here we use a degree-based estimator
by Mcgregor and Vorotnikova \cite{MV16} described
 in the following lemma.
\begin{lemma}
\label{lem:vorotnikova}
Let $G$ be a planar graph. We have $$m(G)\le A'(G)=\sum_{u\in V}\min\{\deg(u)/2,4-\deg(u)/2\}
\le 12.5 \:m(G).$$
\end{lemma}

It is easy to see that, in the simultaneous vertex-partition model,
we can implement this estimator with $O(\log n)$ bits communication
from each player.  

\paragraph*{The final protocol}
Let $k=\lceil n^{1/3} \rceil$. We run the following threads
in parallel.  
\begin{enumerate}
\item A protocol that implements the high-level estimator (Algorithm \ref{alg:locally superior}) with $s=
\lceil 12.5 n/k\rceil$ as its 
input parameter according to the discussions above. 
Let $z_1$ be the output of this protocol.
\item A protocol to compute an instance of $\textrm{Sample}_{b,d,r}$ for 
$b=100k$ and $d=2$ and $r=O(\log k)$. Let $z_2$ be the size of maximum matching in the  
sampled set. 

\item A protocol to compute $A'(G)$. Let $z_3$ be the output of
this thread. 
\end{enumerate}

In the end, if $z_3 \ge  \frac{k}{12.5} $, the referee outputs $z_1$ as an approximation for $m(G)$, 
otherwise the referee reports $z_2$ as the final answer.

\begin{theorem}
Let $G$ be a planar graph on $n$ vertices. The above simultaneous protocol
with probability $3/4$
returns a $3.5+O(\eps)$ approximation
 of $m(G)$ where each player sends $O(\frac{n^{2/3}}{\eps^2})$ bits
 to the referee.
\end{theorem}

\begin{proof}
First we note that by choosing the constants large enough, we can assume the thread (2) errs with probability 
at most $1/8$. 
If $z_3 \ge \frac{k}{12.5} $, then we know $m(G) \ge \frac{k}{12.5} $.
This follows from Lemma \ref{lem:vorotnikova}. Consequently by Lemma
\ref{lem:locally superior}, we have
$\ell(G)\ge \frac{k}{12.5}$. Therefore from Lemma 
\ref{lem:locally superioralg}, we 
have $|z_1-\ell(G)|\le \eps \ell(G)$ with probability at least $7/8$.
It follows from Lemma \ref{lem:planar} that
$(1-\eps)m(G) \le z_1 \le (3.5 +3.5\eps) m(G).$

On the other hand, if $z_3 < \frac{k}{12.5}$,  by Lemma \ref{lem:vorotnikova}
we know that $m(G)$ must be less than $k$. Having this, from Lemma  
\ref{lem:Sample}, with probability at
least $7/8$, we get $z_2 = m(G)$. In this case the protocol 
computes the exact matching size of the graph. 

The communication
complexity each player is dominated by the cost of the first thread
which is $O(n^{2/3}\eps^{-2}\log n)$. 
The total error probability is bounded
by $1/4$. This finishes the proof. 
\end{proof}

\newcommand{\Proc}{Proceedings of the~}

\newcommand{\STOC}{Annual ACM Symposium on Theory of Computing (STOC)}
\newcommand{\FOCS}{IEEE Symposium on Foundations of Computer Science (FOCS)}
\newcommand{\SODA}{Annual ACM-SIAM Symposium on Discrete Algorithms (SODA)}
\newcommand{\SOCG}{Annual Symposium on Computational Geometry (SoCG)}
\newcommand{\ICALP}{Annual International Colloquium on Automata, Languages and Programming (ICALP)}
\newcommand{\ESA}{Annual European Symposium on Algorithms (ESA)}
\newcommand{\CCC}{Annual IEEE Conference on Computational Complexity (CCC)}
\newcommand{\RANDOM}{International Workshop on Randomization and Approximation Techniques in Computer Science (RANDOM)}
\newcommand{\APPROX}{International Workshop on  Approximation Algorithms for Combinatorial Optimization Problems  (APPROX)}
\newcommand{\PODS}{ACM SIGMOD Symposium on Principles of Database Systems (PODS)}
\newcommand{\SSDBM}{ International Conference on Scientific and Statistical Database Management (SSDBM)}
\newcommand{\ALENEX}{Workshop on Algorithm Engineering and Experiments (ALENEX)}
\newcommand{\BEATCS}{Bulletin of the European Association for Theoretical Computer Science (BEATCS)}
\newcommand{\CCCG}{Canadian Conference on Computational Geometry (CCCG)}
\newcommand{\CIAC}{Italian Conference on Algorithms and Complexity (CIAC)}
\newcommand{\COCOON}{Annual International Computing Combinatorics Conference (COCOON)}
\newcommand{\COLT}{Annual Conference on Learning Theory (COLT)}
\newcommand{\COMPGEOM}{Annual ACM Symposium on Computational Geometry}
\newcommand{\DCGEOM}{Discrete \& Computational Geometry}
\newcommand{\DISC}{International Symposium on Distributed Computing (DISC)}
\newcommand{\ECCC}{Electronic Colloquium on Computational Complexity (ECCC)}
\newcommand{\FSTTCS}{Foundations of Software Technology and Theoretical Computer Science (FSTTCS)}
\newcommand{\ICCCN}{IEEE International Conference on Computer Communications and Networks (ICCCN)}
\newcommand{\ICDCS}{International Conference on Distributed Computing Systems (ICDCS)}
\newcommand{\VLDB}{ International Conference on Very Large Data Bases (VLDB)}
\newcommand{\IJCGA}{International Journal of Computational Geometry and Applications}
\newcommand{\INFOCOM}{IEEE INFOCOM}
\newcommand{\IPCO}{International Integer Programming and Combinatorial Optimization Conference (IPCO)}
\newcommand{\ISAAC}{International Symposium on Algorithms and Computation (ISAAC)}
\newcommand{\ISTCS}{Israel Symposium on Theory of Computing and Systems (ISTCS)}
\newcommand{\JACM}{Journal of the ACM}
\newcommand{\LNCS}{Lecture Notes in Computer Science}
\newcommand{\RSA}{Random Structures and Algorithms}
\newcommand{\SPAA}{Annual ACM Symposium on Parallel Algorithms and Architectures (SPAA)}
\newcommand{\STACS}{Annual Symposium on Theoretical Aspects of Computer Science (STACS)}
\newcommand{\SWAT}{Scandinavian Workshop on Algorithm Theory (SWAT)}
\newcommand{\TALG}{ACM Transactions on Algorithms}
\newcommand{\UAI}{Conference on Uncertainty in Artificial Intelligence (UAI)}
\newcommand{\WADS}{Workshop on Algorithms and Data Structures (WADS)}
\newcommand{\SICOMP}{SIAM Journal on Computing}
\newcommand{\JCSS}{Journal of Computer and System Sciences}
\newcommand{\JASIS}{Journal of the American society for information science}
\newcommand{\PMS}{ Philosophical Magazine Series}
\newcommand{\ML}{Machine Learning}
\newcommand{\DCG}{Discrete and Computational Geometry}
\newcommand{\TODS}{ACM Transactions on Database Systems (TODS)}
\newcommand{\PHREV}{Physical Review E}
\newcommand{\NATS}{National Academy of Sciences}
\newcommand{\MPHy}{Reviews of Modern Physics}
\newcommand{\NRG}{Nature Reviews : Genetics}
\newcommand{\BullAMS}{Bulletin (New Series) of the American Mathematical Society}
\newcommand{\AMSM}{The American Mathematical Monthly}
\newcommand{\JAM}{SIAM Journal on Applied Mathematics}
\newcommand{\JDM}{SIAM Journal of  Discrete Math}
\newcommand{\JASM}{Journal of the American Statistical Association}
\newcommand{\AMS}{Annals of Mathematical Statistics}
\newcommand{\JALG}{Journal of Algorithms}
\newcommand{\TIT}{IEEE Transactions on Information Theory}
\newcommand{\CM}{Contemporary Mathematics}
\newcommand{\JC}{Journal of Complexity}
\newcommand{\TSE}{IEEE Transactions on Software Engineering}
\newcommand{\TNDE}{IEEE Transactions on Knowledge and Data Engineering}
\newcommand{\JIC}{Journal Information and Computation}
\newcommand{\ToC}{Theory of Computing}
\newcommand{\MST}{Mathematical Systems Theory}
\newcommand{\Com}{Combinatorica}
\newcommand{\NC}{Neural Computation}
\newcommand{\TAP}{The Annals of Probability}
\newcommand{\TCS}{Theoretical Computer Science}
\newcommand{\IPL}{Information Processing Letter}
\newcommand{\Algorithmica}{Algorithmica}

\bibliographystyle{amsalpha}
  \bibliography{gem}

\newcommand{\etalchar}[1]{$^{#1}$}
\providecommand{\bysame}{\leavevmode\hbox to3em{\hrulefill}\thinspace}
\providecommand{\MR}{\relax\ifhmode\unskip\space\fi MR }
\providecommand{\MRhref}[2]{%
  \href{http://www.ams.org/mathscinet-getitem?mr=#1}{#2}
}
\providecommand{\href}[2]{#2}
\begin{thebibliography}{BGM{\etalchar{+}}19}

\bibitem[AMS99]{alon1999space}
Noga Alon, Yossi Matias, and Mario Szegedy, \emph{The space complexity of
  approximating the frequency moments}, Journal of Computer and system sciences
  \textbf{58} (1999), no.~1, 137--147.

\bibitem[BGM{\etalchar{+}}19]{BuryGMMSVZ19}
Marc Bury, Elena Grigorescu, Andrew McGregor, Morteza Monemizadeh, Chris
  Schwiegelshohn, Sofya Vorotnikova, and Samson Zhou, \emph{Structural results
  on matching estimation with applications to streaming}, Algorithmica
  \textbf{81} (2019), no.~1, 367--392.

\bibitem[CCE{\etalchar{+}}16]{CCEHMMV16}
R.~Chitnis, G.~Cormode, H.~Esfandiari, M.T. Hajiaghayi, A.~McGregor,
  M.~Monemizadeh, and S.~Vorotnikova, \emph{Kernelization via sampling with
  applications to finding matchings and related problems in dynamic graph
  streams}, \Proc 27th \SODA, 2016, pp.~1326--1344.

\bibitem[CJMM17]{CormodeJMM17}
Graham Cormode, Hossein Jowhari, Morteza Monemizadeh, and S.~Muthukrishnan,
  \emph{The sparse awakens: Streaming algorithms for matching size estimation
  in sparse graphs}, 25th Annual European Symposium on Algorithms, {ESA} 2017,
  September 4-6, 2017, Vienna, Austria, 2017, pp.~29:1--29:15.

\bibitem[EHL{\etalchar{+}}15]{EHLMO15}
H.~Esfandiari, M.T. Hajiaghyi, V.~Liaghat, M.~Monemizadeh, and K.~Onak,
  \emph{Streaming algorithms for estimating the matching size in planar graphs
  and beyond}, \Proc 26th \SODA, 2015.

\bibitem[MV16]{MV16}
A.~McGregor and S.~Vorotnikova, \emph{Planar matching in streams revisited},
  \Proc 19th \APPROX, 2016.

\bibitem[MV18]{MV18}
Andrew McGregor and Sofya Vorotnikova, \emph{A simple, space-efficient,
  streaming algorithm for matchings in low arboricity graphs}, 1st Symposium on
  Simplicity in Algorithms, {SOSA} 2018, January 7-10, 2018, New Orleans, LA,
  {USA}, 2018, pp.~14:1--14:4.

\bibitem[NW64]{NW64}
C.~St. J.~A. Nash-Williams, \emph{Decomposition of finite graphs into forests},
  Journal of the London Mathematical Society \textbf{39} (1964), no.~1, 12.

\bibitem[Vit85]{Vitter85}
Jeffrey~Scott Vitter, \emph{Random sampling with a reservoir}, {ACM} Trans.
  Math. Softw. \textbf{11} (1985), no.~1, 37--57.

\end{thebibliography}
%
\end{document}